\documentclass[a4paper,11pt]{article}
\usepackage{amsmath}
\usepackage{amsthm}
\usepackage{amsfonts}
\usepackage{mathbbol}
\usepackage{mathrsfs}
\numberwithin{equation}{section}
\usepackage{setspace}
\usepackage{hyperref}
\hypersetup{hidelinks}
\usepackage{latexsym}
\usepackage{xcolor}
\usepackage{cite}
\usepackage[top=1 in,bottom=1 in,left=1.0 in,right=1.0 in]{geometry}

\newtheorem{proposition}{Proposition}[section]

\newtheorem{remark}{Remark}[section]

\allowdisplaybreaks
			
\title{Nonlinearization of bilinear equations of the sine-Gordon type,
nonlinear Schr\"odinger type and Benjamin-Ono type}

\author{Jin Liu$^{1,2}$,~~ Da-jun Zhang$^{1,2}$\footnote{Corresponding author.
		Email: djzhang@staff.shu.edu.cn},
    ~~ Xin Zhang$^{1}$,
	~~ Xuehui Zhao$^{3,4,5}$\\
{\small $^{1}$Department of Mathematics, Shanghai University, Shanghai 200444, China} \\
{\small $^{2}$Newtouch Center for Mathematics of Shanghai University,  Shanghai 200444, China}\\
{\small $^{3}$College of Mathematical Science, Inner Mongolia Normal University, Hohhot 010022, China}\\
{\small $^{4}$Center for Applied Mathematical Science, Inner Mongolia Normal University, Hohhot 010022, China}\\
{\small $^{5}$Laboratory of Infinite-dimensional Hamiltonian System and Its Algorithm Application,}\\
{\small Inner Mongolia Normal University, Hohhot 010022,  China}}

\date{\today}
			
\begin{document}
				
\maketitle

\begin{abstract}

This  is a continuation of the paper [Commun. Theor. Phys., 77 (2025) 115006]
on the nonlinearization of bilinear equations.
The sine-Gordon type and nonlinear Schr\"odinger type bilinear equations are introduced by
Jarmo Hietarinta during his search for integrable bilinear equations.
In this paper, we provide a formulation to convert these two types of bilinear equations into nonlinear forms.
In addition, the nonlinearization related to the equations involving the Hilbert transformations is also considered.
Bell polynomials are employed in the nonlinearization and illustrative examples are provided.
					
%MSC: 35Q51, 35Q55

\end{abstract}

%\tableofcontents
				
\section{Introduction}\label{sec-1}

Hietarinta searched for integrable bilinear equations from
the Korteweg-de Vries (KdV) type \cite{Hie-1987a},
the modified KdV (mKdV) type \cite{Hie-1987b},
the sine-Gordon (sG) type \cite{Hie-1987c} and the nonlinear Schr\"odinger (NLS) type \cite{Hie-1988}.
For the KdV-type bilinear equations, Hirota found they always have 2-soliton solutions
no matter whether they are integrable or not \cite{H-1980}.
Hirota conjectured having a 3-soliton solution indicates a kind of integrability,
which is later referred to as the integrability for bilinear equations in Hirota's sense.
In Hietarinta's search results, he listed out possible
bilinear equations of the KdV type, the mKdV type and the sG type
that have 3-soliton solutions \cite{Hie-1987a,Hie-1987b,Hie-1987c},
while for the NLS-type bilinear equations which involve complex conjugation,
he pointed out  having a 2-soliton solution might indicate the integrability \cite{Hie-1988}.
In Hietarinta's investigation, he found some new integrable bilinear equations,
but whose nonlinear forms are difficult to write out.
Recently, we developed an approach to convert bilinear equations into nonlinear forms
by utilizing  Bell polynomials \cite{ZLZ-CTP-2025},
by which we are able to nonlinearize the KdV-type and mKdV-type bilinear equations.
The sG type and NLS type are still outstanding.
The difficulties are how to recover trigonometric functions
in the sG type and how to deal with the involved complex conjugate functions in the NLS type.
In this paper, we will provide a formulation to convert these two types of bilinear equations into their nonlinear forms.
In addition, the nonlinearization related to the equations involving the Hilbert transformations
will also be considered.

The paper is organized as follows.
In Section \ref{sec-2}, we recall the basic notions and notations of the Bell polynomials and their
connections with Hirota's bilinear derivatives.
Section \ref{sec-3} and Section \ref{sec-4} illustrate with examples
the procedure to convert the sG-type and NLS-type bilinear equations to their nonlinear forms, respectively.
In Section \ref{sec-5}, we introduce some properties of the Hilbert transformation,
and consider nonlinearizations related to the equations involving the Hilbert transformations.
Finally, concluding remarks  are given in Section \ref{sec-6}.

\section{Preliminary}\label{sec-2}
In this section, we briefly review the Bell polynomials and their connections with
Hirota's bilinear derivatives.
For details one can refer to \cite{GLNW-PRSL-1996,LLSW-JPA-1994}, or refer to Sec.2 of \cite{ZLZ-CTP-2025}.

For $F$ and $G$ being  C$^{\infty}$ functions of $\mathbf{x}=(x_1,x_2,\cdots, x_{\ell})$,
their Hirota's bilinear derivative is defined as \cite{Hir-PTP-1974,Hirota-book}
\begin{equation}\label{D-Hirota}
 D^{n_{1}}_{x_{1}}\cdots D^{n_{\ell}}_{x_{\ell}} F(\mathbf{x})\cdot G(\mathbf{x})
 =(\partial_{x_{1}}-\partial_{x_{1}^{\prime}})^{n_{1}}
 \cdots (\partial_{x_{\ell}}-\partial_{x_{\ell}^{\prime}})^{n_{\ell}}
F(\mathbf{x}) G(\mathbf{x}^\prime)|_{\mathbf{x}^{\prime}=\mathbf{x}},
\end{equation}
where $D$ is called Hirota's bilinear operator.

Let $h=h(x)$ be a $\mathrm{C}^\infty$ function of $x$ and
denote $h_r=\partial^{r}_{x} h$ for $r=1,2,\cdots$.
Then, the so-called 	 Bell polynomials are generated by $h(x)$ are \cite{B-AM-1934}	
\begin{equation}\label{Y-def}
	Y_{[nx]}(h) \doteq Y_n(h_1, h_2, \cdots, h_n) = e^{-h}\partial^{n}_{x}e^{h}, ~~~~~
	n=1,2,\cdots.
\end{equation}
These polynomials can also be generated recursively by
\begin{equation}\label{Y-p-recur}
	Y_0=1, ~~ Y_{n+1}(h_1,h_2,\cdots,h_{n+1})=(\partial_x+h_1)Y_n(h_1,h_2,\cdots,h_n),
\end{equation}
which agrees with the (potential) Burgers hierarchy (see \cite{GLNW-PRSL-1996,LLSW-JPA-1994}).
Hereafter we call the polynomials defined in \eqref{Y-def}
$Y$-polynomials, the first three of which read
\begin{subequations}
\begin{align}
&Y_{[x]} =h_1,\\
&Y_{[2x]} =h_2+h_{1}^{2},\\
&Y_{[3x]} =h_3+3h_1h_2+h_{1}^{3}.
\end{align}
\end{subequations}
Replacing $h(x)$ by $h(-x)$ in \eqref{Y-def}, one can  find that
\begin{align}
Y_{n}(-h_{1},h_{2}, \cdots ,(-1)^{n}h_{n})
=(-1)^{n}Y_{[nx]}(h),\label{-Y}
\end{align}
of which a special case yields
\begin{equation}\label{odd-Y}
	Y_{[(2n+1)x]}(h)|_{h_1=h_3=\cdots=h_{2n+1}=0}=0.
\end{equation}
In addition, using the relation
\begin{equation}
(FG)^{-1}\partial^{n}_{x}(FG)
=\sum_{p=0}^{n}\binom{n}{p}(F^{-1}\partial^{n-p}_{x}F)(G^{-1}\partial^{p}_{x}G),
\end{equation}
and taking $F=e^{h}$, $G=e^{g}$,
one can obtain an addition formula for $Y$-polynomials:
\begin{equation}
Y_{[nx]}(h+g)
=\sum_{p=0}^{n}\binom{n}{p}Y_{[(n-p)x]}(h)Y_{[px]}(g), \label{add-Y}
\end{equation}
where $\binom{n}{p}=\frac{n!}{p!(n-p)!}$.

In 1-dimension case, the relation between Hirota's bilinear derivatives and
Bell polynomials is described as the following \cite{GLNW-PRSL-1996,LLSW-JPA-1994}
\begin{equation}\label{D-Y-2}
(FG)^{-1}D^{n}_{x} F\cdot{G}=
Y_{n}(h_{1},\cdots, h_{n})|_{h_{m}= f_{m}+(-1)^{m}g_{m}},
\end{equation}
where $f_{m}=\partial^m_x f,~ g_{m}=\partial^m_x g$.

We also need to introduce binary deformation of the Bell polynomials:
\begin{equation}\label{Bell-binary}
\mathcal{Y}_{[nx]}(v,u)=Y_{n}(h_{1},\cdots,h_{n})|_{h_{2j-1}=v_{2j-1},\,h_{2j}=u_{2j}}.
\end{equation}
and call them  $\mathcal{Y}$-polynomials,
where $v_{m}=\partial^m_x v,~ u_{m}=\partial^m_x u$.
The first three of them are
\begin{subequations}
\begin{align}
&\mathcal{Y}_{[x]}(v,u)=v_{1},\\
&\mathcal{Y}_{[2x]}(v,u)=u_{2}+v^{2}_{1},\\
&\mathcal{Y}_{[3x]}(v,u)=v_{3}+3u_{2}v_{1}+v^{3}_{1}.
%&\mathcal{Y}_{[4x]}(v,u)=u_{4}+4v_{3}v_{1}+3u^{2}_{2}+6u_{2}v^{2}_{1}+v^{4}_{1}.
\end{align}
\end{subequations}
In terms of the binary Bell polynomials,  the relation \eqref{D-Y-2} is written as
\begin{equation}\label{D-Y-3}
(FG)^{-1}D^{n}_{x} F\cdot G =\mathcal{Y}_{[nx]}(v=\ln (F/G),u=\ln (FG)),
\end{equation}
where  $u=f+g$ and $v=f-g$.
When $G=F$, we have $v=0$, and then in light of \eqref{Bell-binary} and \eqref{odd-Y}
we have
\begin{equation}\label{Y-odd}
	\mathcal{Y}_{[(2j+1)x]}(v=0,u) =0,
\end{equation}
which agrees with the property $D^{2j+1}_x F\cdot F=0$.
We introduce $P$-polynomials for the case of $n=2j$:
\begin{equation}\label{P}
P_{[2jx]}(u) \equiv \mathcal{Y}_{[2jx]}(0,u),
\end{equation}
where the first three read
\begin{subequations}\label{P-polys}
\begin{align}
&P_{[0]}(u) = 1,\\
&P_{[2x]}(u) = u_{2},\\
&P_{[4x]}(u) = u_{4}+3u^{2}_{2}.
%&P_{[6x]}(u) = u_{6}+15u_{2}u_{4}+15u_{2}^{3}.
\end{align}
\end{subequations}
In this case, the bilinear derivative $D^{2j}_{x} F\cdot F$ can be formulated by the $P$-polynomials:
\begin{equation}\label{D-Y-4}
	F^{-2}D^{2j}_{x} F\cdot F =P_{[2jx]}(u=2\ln F).
\end{equation}

Next, we come to the multidimensional generalization.
Assume $h=h(\mathbf{x})$ is an $\ell$-dimension C$^{\infty}$ function
of $\mathbf{x}=(x_1,x_2,\cdots, x_{\ell})$
and denote
\begin{equation}
	h_{r_{1},\cdots,r_{\ell}} =\partial^{r_{1}}_{x_{1}}\cdots\partial^{r_{\ell}}_{x_{\ell}}h(x_{1},\cdots,x_{\ell}).
\end{equation}
The $\ell$-dimension Bell polynomials (also noted as $Y$-polynomials) of $h(\mathbf{x})$
are defined by \cite{GLNW-PRSL-1996,LLSW-JPA-1994}
\begin{equation}\label{Bell-el}
	Y_{[n_{1}x_{1}, \cdots, n_{\ell}x_{\ell}]}(h) \doteq
	Y_{n_{1}, \cdots, n_{\ell}}(\{h_{r_{1}, \cdots, r_{\ell}}\})
	=e^{-h}\partial^{n_{1}}_{x_{1}} \cdots \partial^{n_{\ell}}_{x_{\ell}}e^{h}.
\end{equation}
Examples are the following. When $h=h(x_1,x_2)$, we have
\begin{align*}
& Y_{[x_1,x_2]}=h_{1,0}h_{0,1}+h_{1,1},\\
& Y_{[2x_1,x_2]}=2h_{1,0}h_{1,1}+h_{1,0}^2 h_{0,1}+h_{2,0}h_{0,1}+h_{2,1},\\
& Y_{[3x_1,x_2]}=3h_{1,0}^2h_{1,1}+3h_{1,1} h_{2,0}+3h_{1,0}h_{2,1}
+h_{0,1}(h_{1,0}^3+3h_{1,0}h_{2,0}+h_{3,0})+h_{3,1},
\end{align*}
and when $h=h(x_1,x_2,x_3)$, we have
\[Y_{[x_1,x_2,x_3]}=h_{1,0,0}h_{0,1,1}+h_{1,0,1}h_{0,1,0}+h_{1,1,0}h_{0,0,1}
+h_{1,0,0}h_{0,1,0}h_{0,0,1}+h_{1,1,1}.\]
In addition, from the definition, it holds that
\begin{equation}
	Y_{n_{1}, \cdots, n_{\ell}}(\{(-1)^{r_{1}+ \cdots +r_{\ell}}h_{r_{1}, \cdots, r_{\ell}}\})
	=(-1)^{n_{1}+ \cdots +n_{\ell}}Y_{n_{1}, \cdots, n_{\ell}}(\{h_{r_{1}, \cdots, r_{\ell}}\}).
\end{equation}

The binary Bell polynomials of the  multidimensional case can  be introduced \cite{GLNW-PRSL-1996,LLSW-JPA-1994}:
\begin{subequations}\label{Bell-binary-md}
\begin{equation}
\mathcal{Y}_{[n_{1}x_{1},\cdots,n_{\ell}x_{\ell}]}(v,u) \doteq
Y_{n_{1},\cdots,n_{\ell}}(\{h_{r_{1},...,r_{\ell}}\})
\end{equation}
where
\begin{equation}
h_{r_{1},\cdots,r_{\ell}}=\biggl\{\begin{array}{ll}
u_{r_{1},\cdots,r_{\ell}}, & \mathrm{if}~ r_{1}+\cdots+r_{\ell}  ~\mathrm{is~ even}, \\
v_{r_{1},\cdots,r_{\ell}}, & \mathrm{if }~ r_{1}+\cdots+r_{\ell}  ~\mathrm{is~ odd}.
\end{array}
\end{equation}
\end{subequations}
Again, we provide some examples.
In 2-dimension case, i.e., $\ell=2$ and denoting $x_1=x, x_2=y$, we have
\begin{subequations}
\begin{align}
	&\mathcal{Y}_{[x]}(v,u) =v_{x},\\
	&\mathcal{Y}_{[2x]}(v,u) =u_{2x}+v^{2}_{x},\\
	&\mathcal{Y}_{[x,y]}(v,u) =u_{x,y}+v_{x}v_{y},\\
	&\mathcal{Y}_{[3x]}(v,u) =v_{3x}+3v_{x}u_{2x}+v_{x}^{3},\\
	&\mathcal{Y}_{[2x,y]}(v,u) =v_{2x,y}+2v_{x}u_{x,y}+v^{2}_{x}v_{y}+u_{2x}v_{y},
\end{align}
\end{subequations}
where we use $u_{kx,sy}$ to denote $\partial_x^k\partial_y^s u(x,y)$.
The simplest example for 3 dimensional case ($(x_1,x_2,x_3)=(x,y,z)$) is
\begin{equation}
  \mathcal{Y}_{[x,y,z]}(v,u) =v_{x,y,z}+v_{x}u_{y,z}+v_{y}u_{x,z}+v_{z}u_{x,y}+v_{x}v_{y}v_{z}.
\end{equation}

For the multidimensional binary Bell polynomials, we also have
\[\mathcal{Y}_{[n_{1}x_{1},\cdots,n_{\ell}x_{\ell}]}(v=0,u)=0,~~~
\mathrm{if}~ n_{1}+\cdot\cdot\cdot+n_{\ell}~ \mathrm{is~ odd}.\]
When $n_{1}+\cdots+n_{\ell}$ is even, we denote
\begin{align}
P_{[n_{1}x_{1},\cdots,n_{\ell}x_{\ell}]}(u) =\mathcal{Y}_{[n_{1}x_{1},\cdots,n_{\ell}x_{\ell}]}(v=0,u),
\end{align}
and some examples are
\begin{subequations}
\begin{align}
& P_{[2x]}(u) =u_{2x}, \label{2.24a}\\
& P_{[x,y]}(u) =u_{x,y}, \label{2.24b}\\
& P_{[2x,2y]}(u) =2u^2_{x,y}+u_{2x} u_{2y}+u_{2x,2y},\\
& P_{[3x,y]}(u) =u_{3x,y}+3u_{2x}u_{x,y}.
\end{align}
\end{subequations}

In the multidimensional case, the addition formula reads
\begin{equation}
	Y_{[n_{1}x_{1}, \cdots, n_{\ell}x_{\ell}]}(f+g)
	=\sum_{p_{1}=0}^{n_{1}}\cdots \sum_{p_{\ell}=0}^{n_{\ell}}
	\prod_{i=1}^{\ell}\binom{n_{i}}{p_{i}}Y_{[(n_{1}-p_{1})x_{1}, \cdots, (n_{\ell}-p_{\ell})x_{\ell}]}(f)
	Y_{[p_{1}x_{1}, \cdots, p_{\ell}x_{\ell}]}(g).
\end{equation}
Introducing  $F = e^{f(\mathbf{x})} $ and $G =e^{g(\mathbf{x})}$,
the Hirota bilinear derivatives defined in \eqref{D-Hirota} can be formulated as \cite{GLNW-PRSL-1996,LLSW-JPA-1994}
\begin{equation}\label{D-Y-5}
	(FG)^{-1}D_{x_{1}}^{n_{1}}\cdots D_{x_{\ell}}^{n_{\ell}}F \cdot G
	=Y_{n_{1},\cdots,n_{\ell}}(\{h_{r_{1},\cdots,r_{\ell}}=f_{r_{1},\cdots,r_{\ell}}+(-1)^{r_{1}
		+\cdots+r_{\ell}}g_{r_{1},\cdots, r_{\ell}}\}).
\end{equation}
and also through binary Bell polynomials (taking $u = f+g$ and $v=f-g$):
\begin{equation}\label{trans-formula}
	(FG)^{-1}D^{n_{1}}_{x_{1}}\cdots D^{n_{\ell}}_{x_{\ell}} F\cdot G
	=\mathcal{Y}_{[n_{1}x_{1},...,n_{\ell}x_{\ell}]}(v=\ln (F/G), u=\ln (FG)),
\end{equation}
and in particular, when $F=G$, we have
\begin{align}\label{2.30}
	F^{-2}D^{n_{1}}_{x_{1}}\cdots  D^{n_{\ell}}_{x_{\ell}}F \cdot F=
	P_{[n_{1}x_{1},...,n_{\ell}x_{\ell}]}(u=2\ln F)
\end{align}
where $n_{1}+\cdots +n_{\ell}$ is even.
These formulae have been used to convert Hirota's bilinear equations into nonlinear forms
for the KdV type and mKdV type \cite{ZLZ-CTP-2025}.
In the following, we will work on the nonlinearization of the sG-type bilinear equations
and the ones involving complex variables.

\section{The sG-type bilinear equations and nonlinearization}\label{sec-3}

In this section, we will first recall Hietarinta's study of the sG-type bilinear equations in \cite{Hie-1987c}.
Then we will illustrate how trigonometric functions come up with the nonlinearization of this type of bilinear equations.

\subsection{The sG-type bilinear equations}\label{sec-3-1}

In 1987, Hietarinta \cite{Hie-1987c} considered the following type of coupled bilinear equations
(named as the sG-type):
\begin{subequations}\label{sG-bilinear-0}
\begin{align}
&A_{1}(D_{x},D_{t})F\cdot F+B_{1}(D_{x},D_{t})F\cdot G+C_{1}(D_{x},D_{t})G\cdot G=0,\\
&A_{2}(D_{x},D_{t})F\cdot F+B_{2}(D_{x},D_{t})F\cdot G+C_{2}(D_{x},D_{t})G\cdot G=0,
\end{align}
\end{subequations}
where $A_i$ and $C_{i}$ ($i=1,2$) are some binary even polynomials. The above system has
a 1-soliton solution of the following form when the polynomials $B_1$ and $B_2$ are proportional
(including the case that one of the $B_i$ is 0, without loss of generality, we set $B_1 \neq 0$):
\begin{subequations}\label{sg-1ss}
\begin{align}
&F=1,\\
&G=e^{\eta}, ~~ \eta=px+\Omega t+\eta^{(0)},
\end{align}
\end{subequations}
where $p,\Omega,\eta^{(0)}\in \mathbb{C}$, $p$ and $\Omega$ are subject to the dispersion relation:
\begin{equation}
B_{1}(p,\Omega)=0.
\end{equation}

When $B_{1}$ and $B_{2}$ are proportional and $B_{1}\neq0$, Hietarinta looked for the 2-soliton solution of
the following form
\begin{subequations}\label{sg-2ss-0}
\begin{align}
&F=1+A_{12}e^{\eta_{1}+\eta_{2}},\\
&G=e^{\eta_{1}}+e^{\eta_{2}}, ~~ \eta_{i}=p_{i}x+\Omega_{i}t+\eta_{i}^{(0)},\label{def-eta}
\end{align}
\end{subequations}
where $p_{i},\Omega_{i},\eta_{i}^{(0)}\in \mathbb{C}$, $i=1,2$,
$A_{12}$ is to be determined, and $\eta_{i}$ satisfies the dispersion relation
\begin{equation}\label{sg-dp-2-1}
B_{1}(p_{i},\Omega_{i})=0, ~~ i=1,2.
\end{equation}
It turns out that, to make \eqref{sg-2ss-0} to be a solution to \eqref{sG-bilinear-0},
apart from \eqref{sg-dp-2-1}, $A_{i}$ and $B_{i}$ still need to satisfy the following conditions \cite{Hie-1987c}:
\begin{subequations}\label{sg-2ss}
\begin{align}
&A_{12}A_{i}(p_{1}+p_{2},\Omega_{1}+\Omega_{2})+C_{i}(p_{1}-p_{2},\Omega_{1}-\Omega_{2})=0,\\
&B_{i}(-p_{j},-\Omega_{j})=0.
\end{align}
\end{subequations}
where $i,j=1,2$.

To sum up, if the sG-type bilinear system \eqref{sG-bilinear-0} satisfies the following conditions, then it has
the 1 and 2-soliton solutions of the form like \eqref{sg-1ss} and \eqref{sg-2ss} \cite{Hie-1987c}:
\begin{itemize}
\item
$B_{1}$ and $B_{2}$ are proportional, $B_{1}\neq 0$ and
\begin{equation}\label{sg-dp-2}
B_{1}(p_{j},\Omega_{j})=B_{1}(-p_{j},-\Omega_{j})=0, ~~ j=1,2;
\end{equation}
\item
$A_{i}$ and $C_{i}$ satisfy (see equation (24) in \cite{Hie-1987c})
\begin{equation}
A_{1}(p_{1}+p_{2},\Omega_{1}+\Omega_{2})C_{2}(p_{1}-p_{2},\Omega_{1}-\Omega_{2})
=A_{2}(p_{1}+p_{2},\Omega_{1}+\Omega_{2})C_{1}(p_{1}-p_{2},\Omega_{1}-\Omega_{2});
\end{equation}
\item assume
$A_{2}\neq 0$, then
\begin{equation}\label{def-A}
A_{12}=-\frac{C_{2}(p_{1}-p_{2},\Omega_{1}-\Omega_{2})}{A_{2}(p_{1}+p_{2},\Omega_{1}+\Omega_{2})}.
\end{equation}
\end{itemize}

Under the existence of a 2-soliton solution, Hietarinta studied the sG-type equations
which may have 3-soliton solutions of the following type
\begin{subequations}
\begin{align}
&F=1+A_{12} e^{\eta_1 +\eta_2} +A_{13} e^{\eta_1 +\eta_3} +A_{23} e^{\eta_2 +\eta_3},\\
&G=e^{\eta_1} +e^{\eta_2} +e^{\eta_3} +A_{12}A_{13}A_{23} e^{\eta_1 +\eta_2 +\eta_3},
\end{align}
\end{subequations}
where $\eta_{i}=p_{i}x+\Omega_{i}t+\eta_{i}^{(0)}$ satisfying the dispersion relation \eqref{sg-dp-2},
and $A_{ij}$ is defined as in \eqref{def-A}, i.e.,
\begin{equation}
A_{ij}=-\frac{C_{2}(p_{i}-p_{j},\Omega_{i}-\Omega_{j})}{A_{2}(p_{i}+p_{j},\Omega_{i}+\Omega_{j})}.
\end{equation}
The main results of the sG-type bilinear equations possessing 3-soliton solutions are as follows:
\begin{itemize}
\item For $A_{2}(D_{x},D_{t})=-C_{2}(D_{x},D_{t})=D_{x}^{2},~A_{1}=C_{1}=0,~B_{2}=0$,
the possible types of $B_{1}$ are listed in Table I (Accepted final result) in \cite{Hie-1987c}.
\item For $A_{2}(D_{x},D_{t})=-C_{2}(D_{x},D_{t})=D_{x}D_{t},~A_{1}=C_{1}=0,~B_{2}=0$,
the possible types of $B_{1}$ are listed in Table II (Accepted final result) in \cite{Hie-1987c}.
\item
For a gernal case, see Eq.(56) in \cite{Hie-1987c}.
\end{itemize}

\subsection{Nonlinearization}\label{sec-3-2}

In this section, by some  examples we show how nonlinear equations are recovered from the sG-type bilinear equations.

The first one is a toy example of the sG-type bilinear equations (see Table I (Accepted final result) of \cite{Hie-1987c})
\begin{subequations}\label{sg-1}
\begin{align}
&D_x^2  (F \cdot F - G \cdot G) =0,\label{sg-1-a}\\
&(D_x^2 + 1)  F \cdot G =0.\label{sg-1-b}
\end{align}
\end{subequations}
Using the formula \eqref{2.30}
and the transformations
\begin{equation}
v=\ln (F/G),~~ u=\ln(FG),
\end{equation}
for equation \eqref{sg-1-a} we have
\begin{align*}
(F G) ^{-2}D_x^2  (F \cdot F - G \cdot G)
=&G^{-2} (F^{-2} D_x^2 F \cdot F) - F^{-2} (G^{-2}D_x^2 G \cdot G)\\
=&G^{-2} P_{[2x]}(2 \ln F) - F^{-2} P_{[2x]}(2 \ln G).
\end{align*}
From \eqref{2.24a} and noticing that $2\ln F=u+v,~ 2\ln G=u-v$, we arrive at
\begin{equation}\label{sg-1-a1}
e^{-(u-v)} (u+v)_{xx} -  e^{-(u+v)} (u-v)_{xx}=0,
\end{equation}
i.e.
\begin{equation}\label{u_xx v_xx}
u_{xx} =-\frac{ e^v + e ^{-v} }{e^v - e ^{-v}} v_{xx} = -v_{xx} \coth (v) .
\end{equation}
Similarly, for \eqref{sg-1-b}, using \eqref{trans-formula}  we find
\begin{equation}\label{sg-1-b1}
(F G)^{-1} (D_x^2 + 1) F \cdot G=\mathcal{Y}_{[2x]}(v,u) + 1 =u_{xx} + v_x^2 + 1 = 0.
\end{equation}
Substituting \eqref{u_xx v_xx} into \eqref{sg-1-b1}  we have
\begin{equation}
	- v_{xx}\coth(v) + v_x^2 +1 = 0.
\end{equation}
Then we introduce $w$ by applying a variable transformation
\begin{equation}\label{v-w-1}
v = \ln \bigg( \tan \frac{w}{4} \bigg),
\end{equation}
we obtain
\begin{equation}\label{w_xx=-sin w}
w_{xx} = - \sin w,
\end{equation}
where
\begin{equation}\label{sg-tr}
w=4\arctan(e^v)=4\arctan(F/G).
\end{equation}
It is notable  that \eqref{w_xx=-sin w} has a three-soliton solution \cite{Hie-1987c}
while it is an ordinary differential equation. Similar examples are found
in other one-dimensional cases (e.g. Table I, II, III in \cite{Hie-1987a}),
which means that ordinary differential equations may possess multi-soliton solutions.

The second example is the standard bilinear sG equation:
\begin{subequations}\label{sg-bi}
\begin{align}
& D_x D_t (F \cdot F - G \cdot G) =0,\label{sg-a}\\
& (D_x D_t - 1) F \cdot G =0.\label{sg-b}
\end{align}
\end{subequations}
Using the formula \eqref{trans-formula}, from \eqref{sg-a} we have
\begin{equation*}
\begin{aligned}
(F  G)^{-2}
D_x D_t  (F \cdot F - G \cdot G)
=&G^{-2} P_{[x,t]}(2 \ln F) - F^{-2} P_{[x,]t}(2 \ln G)\\
=&e^{-u} [ (e^v - e^{-v} ) u_{xt} + (e^v + e^{-v} ) v_{xt}  ]=0,
\end{aligned}
\end{equation*}
where the polynomial \eqref{2.24b} is used.
It then yields
\begin{equation}\label{sg-a-2}
u_{xt} =-\frac{ e^v + e ^{-v} }{e^v - e ^{-v}} v_{xt}
= - v_{xt}\coth (v) .
\end{equation}
Similarly, \eqref{sg-b} yields
\begin{equation}\label{sg-b-1}
(F G)^{-1}
(D_x D_t - 1) F \cdot G=\mathcal{Y}_{[x,t]}(v,u)-1=u_{xt} + v_x v_t -1 = 0.
\end{equation}
Substituting \eqref{sg-a-2} into \eqref{sg-b-1}  we have
\begin{equation}\label{sg-non-1}
- v_{xt}\coth(v) + v_x v_t -1 = 0.
\end{equation}
Introducing variable $w$ as \eqref{v-w-1},
the above equation becomes the sG equation
\begin{equation}\label{sg-eq}
w_{xt} = \sin w,
\end{equation}
and the transformation between $w$ and $F$ and $G$ takes \eqref{sg-tr}.
In the above procedure, the transformation \eqref{v-w-1} plays a key role in the formulation
of the sine function.

For the last example we  consider a high dimensional case of the sG type equations
(see \cite{Hie-1987c}, Table II, the column ``Generalizations with $Y$''):
\begin{subequations}\label{sg-3-bi}
\begin{align}
&(D_x D_y +1 ) F \cdot G =0,\label{sg-3-b}\\
&D_x D_t  (F \cdot F - G \cdot G) =0.\label{sg-3-a}
\end{align}
\end{subequations}
The second equation gives rise to \eqref{sg-a-2}, while for the first one we have
\begin{subequations}
\begin{align}
(F G)^{-1}
(D_x D_y +1 ) F \cdot G
&= \mathcal{Y}_{[x,y]}(v,u) + 1 \nonumber \\
&=u_{xy} + v_x v_y +1 =0.\label{sg-3-b-2}
\end{align}
\end{subequations}
One can differentiate \eqref{sg-a-2} with respect to $y$ and \eqref{sg-3-b-2} with respect to $t$,
and then eliminate $u_{xyt}$ from them, to obtain an equation
\begin{equation}\label{w_xyt}
w_{xyt} = w_x w_{xy}\cot w -  w_y w_{xt}\tan w,
\end{equation}
where we also introduced $w$ by \eqref{v-w-1}.
The above equation has been presented in \cite{Lou-2003} (see equation (28) in \cite{Lou-2003})
as  a reduction related to the (2+1) dimension sG system found by
Konopelchenko and Rogers \cite{KR-PLA-1991}:
\begin{subequations}\label{sG-2+1}
\begin{align}
&\left(\frac{\phi_\xi}{\sin \theta}\right)_{\xi} -\left(\frac{\phi_\mu}{\sin \theta}\right)_{\mu}
+\frac{\phi_{\mu}\theta_{\xi}-\phi_{\xi}\theta_{\mu}}{\sin^2\theta}=0,\\
&\left(\frac{\psi_\xi}{\sin \theta}\right)_{\xi} -\left(\frac{\psi_\mu}{\sin \theta}\right)_{\mu}
+\frac{\psi_{\mu}\theta_{\xi}-\psi_{\xi}\theta_{\mu}}{\sin^2\theta}=0,
\end{align}
\end{subequations}
where both $\phi$ and $\psi$ are functions of $(\xi, \mu, t)$ and $\theta_t=\phi-\psi$.
This system has received intensively attention from Darboux transformation \cite{N-PLA-1992,N-PANRW-1993,S-JPA-1992}
and dressing method \cite{DK-IP-1993,KD-SAMP-1993}, right after it was proposed.
There is a reduction of \eqref{sG-2+1} due to taking $\psi=0$,
which yields a scalar (2+1) dimension sG equation in terms of $\theta$ \cite{KR-PLA-1991}:
\begin{align}
\left(\frac{\theta_{\xi t}}{\sin \theta}\right)_{\xi} -\left(\frac{\theta_{\mu t}}{\sin \theta}\right)_{\mu}
+\frac{\theta_{\mu t}\theta_{\xi}-\theta_{\xi t}\theta_{\mu}}{\sin^2\theta}=0.
\end{align}
The above equation gives rise to equation \eqref{w_xyt} after defining
\begin{equation}
w(x,y,t)=2\theta(\xi=x+y, \mu=x-y,t).
\end{equation}

\section{The NLS-type bilinear equations and nonlinearization}\label{sec-4}

In this section, we will show the nonlinearization of the NLS-type bilinear equations.
The NLS-type bilinear equations considered in \cite{Hie-1988} by Hietarinta are the following:
\begin{subequations}\label{4.1}
\begin{align}
& B(D_x,D_t) G \cdot F=0,\\
& A(D_x,D_t) F \cdot F= C(D_x,D_t) G \cdot G^*,
\end{align}
\end{subequations}
where $*$ denotes  complex conjugate, the binary polynomials $A, B, C$ satisfy
\begin{subequations}
\begin{align}
&A(0,0)=0,\\
&[B(X,T)]^* = B(-X^*, -T^*),\label{cond b}\\
&[C(X,T)]^* = C(-X^*, -T^*),
\end{align}
\end{subequations}
and in addition, $A$ is assumed to be even with real coefficients.
Bilinear equations in this type automatically possess  a 1-soliton solution  in the following form \cite{Hie-1988}:
\begin{equation}
F = 1 + K e^{\eta + \eta^*}, \quad G = e^\eta,
\end{equation}
where
\begin{equation}
\eta = p x + \Omega t +  \eta^{(0)},~~~  p,\Omega,\eta^{(0)}\in \mathbb{C},
\end{equation}
satisfying the dispersion relation
\begin{equation}
B(p,\Omega)=0,
\end{equation}
and $K$ being defined as
\begin{equation}
K = \frac{C(p - p^*, \Omega - \Omega^*)}{2A(p + p^*, \Omega + \Omega^*)}.
\end{equation}
Note that in this type the bilinear equations do not always have 2-soliton solutions.
In \cite{Hie-1988} Hietarinta searched the above NLS-type bilinear quations that allow
having 2-soliton solutions, and he claimed that having 2-soliton solutions implies integrability for this type of
bilinear equations.

In the following we employ some examples to show how to recover a nonlinear form from
a set of NLS-type bilinear equations.
The first example is the standard bilinear NLS  equations  (cf.\cite{Hir-JMP-1973})
\begin{subequations}\label{nls-1}
\begin{align}
(D_x^2 + i D_t) G \cdot F = 0,\label{nls-1-a}\\
D_x^2 F \cdot F = G G^*,\label{nls-1-b}
\end{align}
\end{subequations}
where $i^2=-1$, and we assume $F$ is a  real function and $G$ is complex.
Introduce
\begin{equation}\label{v,u-3}
v=\ln(F/G),~u=\ln(FG).
\end{equation}
Using  formula \eqref{trans-formula},
from \eqref{nls-1-a} we have
\begin{align}
(FG)^{-1}(D_x^2 + i D_t) G \cdot F
&=\mathcal{Y}_{[2x]}(v,u) - i \mathcal{Y}_{[t]}(v,u) \nonumber\\
&=u_{xx} + v_x^2 - iv_t=0.
\label{nls-1-a-1}
\end{align}
For \eqref{nls-1-b}, we rewrite it as
\[ F^{-2} D_x^2 F \cdot F =\frac{G G^*}{F^2}.\]
It then follows from \eqref{D-Y-4} that
\[ P_{[2x]}(u+v) =\bigg( \frac{G}{F} \bigg) \bigg( \frac{G}{F} \bigg)^*,\]
i.e.
\begin{equation}\label{nls-1-b-1}
(u+v)_{xx}
=e^{-v} (e^{-v})^*.
\end{equation}
Eliminating $u_{xx}$ from   \eqref{nls-1-b-1} and \eqref{nls-1-a-1}, we get
\begin{equation}\label{nls-bell}
iv_t + v_{xx} - v_x^2- e^{-v} (e^{-v})^*   = 0.
\end{equation}
Now we introduce
\begin{equation}
w=e^{-v},
\end{equation}
and then we can get an equation in terms of $w$:
\begin{equation}\label{nls-eq}
i w_t + w_{xx} + |w|^2 w = 0,
\end{equation}
which is nothing but the NLS equation, where $|w|^2=ww^*$.

The next example is
\begin{subequations}\label{4.14}
\begin{align}
&(i a D_x^3 + D_x D_t + i D_y +b ) G \cdot F = 0,\\
&D_x^2 F \cdot F = G G^*,
\end{align}
\end{subequations}
where $F$ is real and $G$ is complex, $a$ and $b$ are real parameters.
This is listed  in the abstract of \cite{Hie-1988} as a new NLS-type bilinear equation
that has a 2-soliton solution.
Under the transformation \eqref{v,u-3} and using \eqref{trans-formula} and \eqref{D-Y-4},
we have
\begin{align}
& (FG)^{-1}(i a D_x^3 + D_x D_t + i D_y +b ) G \cdot F \nonumber \\
= ~& - i a \mathcal{Y}_{[3x]} (v,u) + \mathcal{Y}_{[x,t]} (v,u)
- i \mathcal{Y}_{[y]} (v,u) + b \nonumber \\
=~ &-i a (v_{3x} + 3 v_x u_{2x} + v_x^3 )+ u_{xt} +v_x v_t -i v_y + b
 = 0, \label{nls-3-bi-1}
\end{align}
and
\begin{equation}\label{nls-3-bi-2}
P_{[2x]} (u+v) = (u+v)_{xx} =\frac{GG^*}{F^2}.
\end{equation}
Introducing
\begin{equation}
\psi = e^{-v},
\end{equation}
Eq.\eqref{nls-3-bi-2} yields
\begin{equation}\label{u-psi-1}
u_{xx} =- v_{xx} + | \psi |^2.
\end{equation}
Differentiating \eqref{u-psi-1} with respect to $t$
and introducing $w$ by
\begin{equation}\label{w-3}
w=u_{xt}+v_{xt},
\end{equation}
then, from \eqref{nls-3-bi-1} and \eqref{u-psi-1} we obtain
\begin{subequations}\label{twist-eq}
\begin{align}
&i a ( \psi_{xxx} +3 \psi_x | \psi |^2) + \psi_{xt} + w \psi + i \psi_y + b \psi =0,\\
&w_x = ( | \psi |^2 )_t,
\end{align}
\end{subequations}
which connects with its bilinear form \eqref{4.14} via
\begin{equation}
\psi=\frac{G}{F}, ~~ w = \frac{2D_{x}D_t F \cdot F}{F^2}.
\end{equation}
If $a=b=0$, \eqref{twist-eq} reduces to
\begin{equation}
i \psi_y + \psi_{xt} + w \psi =0, ~\quad w_x = ( | \psi |^2 )_t,
\end{equation}
which is known as the ``breaking soliton NLS equation''
(cf.\cite{LZ-JPA-1993,WWZ-CTP-2020,ZZJ-CPL-2010}).
Zakharov first presented the Lax pair of this equation \cite{Zakh-1980},
and therefore it is also known as the integrable Zakharov equation.

The third example  is
\begin{subequations}\label{zakharov-bi}
	\begin{align}
		&(i D_t + D_x^2) G \cdot F =0,\\
		&(D_x^2 - D_t^2) F \cdot F = G G^*,
	\end{align}
\end{subequations}
where $F$ is real and $G$ is complex.
With  \eqref{v,u-3}, \eqref{trans-formula} and \eqref{D-Y-4},
we have
\begin{subequations}\label{zakharov-bi-1}
\begin{align}
&-i \mathcal{Y}_{[t]} (v,u) + \mathcal{Y}_{[2x]} (v,u)
=-iv_t + u_{xx} + v_x^2 = 0,\\
& P_{[2x]} (q)  - P_{[2t]} (q)
=q_{xx} - q_{tt} = \frac{G G^*}{F^2},
\end{align}
\end{subequations}
in which $q=2\ln F=u+v$.
Introducing
\begin{equation}
E=e^{-v}, ~~\rho=-(u+v)_{xx}=-q_{xx},
\end{equation}
we have
\begin{subequations}\label{zakharov-eq}
\begin{align}
& i E_t + E_{xx} = \rho E,\\
& \rho_{tt} - \rho_{xx} = | E^2 |_{xx}.
\end{align}
\end{subequations}
Eq.\eqref{zakharov-eq} is also called the Zakharov equation
derived by Zakharov  in \cite{Zakh-1972},
but it has been shown non-integrable \cite{Hie-1988,BB-PLA-1983,Herr-JPA-1983}.

There exists another way to derive nonlinear forms through bilinear equations
without dealing with the complex conjugations directly.
Let us take the bilinear NLS system \eqref{nls-1} as an example to illustrate the procedure.
In stead of considering \eqref{nls-1}, we started from its unreduced form
\begin{subequations}\label{akns-bi}
\begin{align}
&(i D_t + D_x^2) G \cdot F =0,\\
&(i D_t - D_x^2) H \cdot F =0,\\
&D_x^2 F \cdot F = -G H,
\end{align}
\end{subequations}
where at this stage  it is not necessary to require $F, G, H$ to be real or complex.
Eq.\eqref{nls-1} can be considered as a special case of \eqref{akns-bi} under the reduction $F=F^*$ and $G=-H^*$.
Under the transformation \eqref{v,u-3} and using \eqref{trans-formula} and  \eqref{D-Y-4} we have
\begin{subequations}\label{akns-bi-1}
\begin{align}
&-i \mathcal{Y}_{[t]} (v,u) + \mathcal{Y}_{[2x]} (v,u)
=-iv_t + u_{xx} + v_x^2 = 0,\\
&-i \mathcal{Y}_{[t]} (w,s) - \mathcal{Y}_{[2x]} (w,s)
=-iw_t - s_{xx} - w_x^2 = 0,\\
&P_{[2x]} (2\ln F)=(2\ln F)_{xx} = -\frac{GH}{F^2}=-e^{-v-w}, \label{4.28c}
\end{align}
\end{subequations}
where
\begin{equation}
v=\ln \frac{F}{G},~~u=\ln FG,~~ w=\ln \frac{F}{H}, ~~s=\ln HF.
\end{equation}
Noticing that $u,v,w,s$ are related by
\begin{equation}
 u+v = s+w = 2\ln F=q,
\end{equation}
from \eqref{4.28c} we have
\begin{subequations}
\begin{align}
&u_{xx}=q_{xx}-v_{xx}=-v_{xx}-e^{-v-w},\\
&s_{xx}=q_{xx}-w_{xx}=-w_{xx}-e^{-v-w}.
\end{align}
\end{subequations}
Then \eqref{akns-bi-1} can be written as
\begin{subequations}
\begin{align}
&i v_t + v_{xx} + e^{-v-w} - v_x^2 =0,\\
&i w_t -w_{xx} - e^{-v-w} +w_x^2=0.
\end{align}
\end{subequations}
Introducing new variables
\begin{equation}
p=e^{-v}, ~~~r=e^{-w},
\end{equation}
we arrive at the (second order) Ablowitz-Kaup-Newell-Segur (AKNS) system
\begin{subequations}\label{akns-eq}
\begin{align}
&i p_t + p_{xx} -p^2 r = 0,\\
&i r_t -r_{xx} + p r^2 = 0.
\end{align}
\end{subequations}
It is connected with the bilinear form \eqref{akns-bi} via $p=G/F, r=H/F$,
and \eqref{nls-eq} can be derived  by taking $r=-p^*$ in \eqref{akns-eq}.
It admits more freedom to study the unreduced system \eqref{akns-bi}
than the reduced one \eqref{nls-1}, in particular in finding solutions for
nonlocal equations, e.g. \cite{CDLZ-SAPM-2018,DCCZ-SAPM-2025}.

\section{The BO-type bilinear equations and nonlinearizations}\label{sec-5}

In this section, we will  provide examples of the nonlinaerization towards  the Benjamin-Ono (BO) type equations
that involve the Hilbert transformations.
First, let us briefly recall the Hilbert transformation and its properties.
For details and  further introduction one can refer to \cite{King-book}.

\subsection{Hilbert transformation}\label{sec-5-1}

The Hilbert transformation of a given function $f(x),~x\in \mathbb{R}$, is defined as
\begin{equation}\label{H}
	\mathcal{H} [f(x)] = \frac{1}{\pi} \mathrm{P.V.} \int_{-\infty}^{+\infty} \frac{f(s)}{s-x} \mathrm{d}s
= \frac{1}{\pi} \lim_{\epsilon\to 0^+}\left( \int_{-\infty}^{x-\epsilon}+
 \int_{x+\epsilon}^{+\infty} \right)
 \frac{f(s)}{s-x} \mathrm{d}s,
\end{equation}
where P.V. denotes the Cauchy principal value of the integral.

The Hilbert transform has the following properties that play important roles in the variable transformations.

\begin{proposition}\label{prop-4-1}\cite{King-book}
If both $f(x)$ and $f_x(x)$ are
functions in $L^p(\mathbb{R})$ space where $1< p <+\infty$,
then the transformation $\mathcal{H}$ and differentiation $\partial_x$ commute,
i.e.
\begin{equation}
	\partial_x \mathcal{H} [f(x)]=
	\mathcal{H} [\partial_x f(x)].
\end{equation}
\end{proposition}
\begin{proposition}\label{prop-4-2}\cite{King-book}
If function $f(z),~z\in \mathbb{C}$, is analytic in the upper half plane
and $f(z) \to 0$ for $\mathrm{Im}\, z \geq 0$ as $|z| \to \infty$,
then for any $x  \in \mathbb{R}$, there is
\begin{equation}
\mathcal{H}[f(x )] = i f(x ).
\end{equation}
Similarly, for the function $g(z)$ analytic in the lower half plane, and $g(z) \to 0$
for $\mathrm{Im}\, z \leq 0$  as $|z| \to \infty$, we have
\begin{equation}\label{lem-lower-half-plane}
\mathcal{H}[g(x)] = -i g(x), ~~ x  \in \mathbb{R}.
\end{equation}
\end{proposition}

\begin{proposition}\label{prop-4-3}
 Consider two  functions $F(z), G(z), z\in \mathbb{C}$,
which are analytic in the upper half and lower half complex plane, respectively.
If their zeros are located  in the lower half and upper half complex plane, respectively,
i.e.
\begin{subequations}\label{cond-hilbert-trans-formula}
\begin{align}
& F(s_j)=0, ~ ~\mathrm{Im}\, s_j  <0, ~~
(j=1,2,\cdots, M),\\
& G(z_j)=0, ~ ~\mathrm{Im}\, z_j  >0, ~~
(j=1,2,\cdots, N),
\end{align}
\end{subequations}
and they also satisfy $\partial_x\ln F \to 0$ (for $\mathrm{Im}\, z \geq 0$)
and $\partial_x\ln G \to 0$ (for $\mathrm{Im}\, z \leq 0$) as $z$ goes to infinity,
then it holds that
\begin{equation}\label{hil-eq1}
i \mathcal{H} \bigg[\partial_x \ln \frac{F}{G} \bigg]
= - \partial_x  \ln (FG),
\end{equation}
where $x=\mathrm{Re}(z)$.
Note that $\partial_x F(z)=\partial_z F(z)$ holds in the upper half complex plane and and be extended to the real axis.
Similarly, $\partial_x G(z)=\partial_z G(z)$ holds as well on the real axis.
\end{proposition}

\begin{proof}
Consider two functions
\begin{equation}
f(z) = \partial_x \ln F(z)
= \frac{F_x}{F},
\end{equation}
\begin{equation}
g(z) = \partial_x \ln G(z)
= \frac{G_x}{G}.
\end{equation}
Since all the zeros of $F$ are in the lower half plane,
$f$ is analytic in the upper half plane and satisfies the condition of Proposition \ref{prop-4-2}.
Thus it follows that
\begin{equation}
\mathcal{H}[ f(x)] = i f(x), ~~ x \in \mathbb{R}.
\end{equation}
Similarly, there is
\begin{equation}
\mathcal{H} [g(x)] = -i g(x), ~~ x \in \mathbb{R}.
\end{equation}
Then we have
\begin{equation}
i \mathcal{H} \bigg[\partial_x\ln \frac{F}{G} \bigg]
= i \mathcal{H} [f-g]
=  i (i f(x) + i g(x) )
= -f -  g
= - \partial_x  \ln (FG),
\end{equation}
which is \eqref{hil-eq1}.
\end{proof}

Combining the formulae \eqref{hil-eq1} and  \eqref{trans-formula}, we immediately find the following.

\begin{proposition}\label{prop-4-4}
For the functions $F(z)$ and $G(z)$ satisfying the conditions in Proposition \ref{prop-4-3},
introducing functions
\begin{equation}\label{uv-FG}
v=\ln(F/G), ~~ u=\ln(FG),
\end{equation}
from formula \eqref{hil-eq1} we have
\begin{equation}\label{u_x=-iHv_x}
u_x =-  i \mathcal{H}[v_x].
\end{equation}
Then, assuming both $F$ and $G$ are functions of $(x_1=x, x_2,\cdots, x_l)$,
in light of formula \eqref{trans-formula}, we obtain
\begin{align}
(FG)^{-1}D^{n_{1}}_{x_{1}}\cdots D^{n_{\ell}}_{x_{\ell}} F\cdot G
&=\mathcal{Y}_{[n_{1}x_{1},...,n_{\ell}x_{\ell}]}(v=\ln(F/G), u=\ln(FG)) \nonumber \\
&=\mathcal{Y}_{[n_{1}x_{1},...,n_{\ell}x_{\ell}]}
(v,  -i\int \mathcal{H}[v_x] \mathrm{d}x),
\label{4.51}
\end{align}
which connects bilinear derivatives with the function $v$ and the related Hilbert transformation.
In particular, if  $v, v_x \in L^p(\mathbb{R})$\footnote{
In this paper when say $f(z)\in L^p(\mathbb{R})$
we mean both $\mathrm{Re} f(z=x)$ and $\mathrm{Im} f(z=x)$  belong to $L^p(\mathbb{R})$.
 }
where $1< p <+\infty$, then, in light of Proposition \ref{prop-4-1},
 \eqref{u_x=-iHv_x}
yields $u=-i \mathcal{H}[v]+c$
where $c$ is a constant.
\end{proposition}

\subsection{Nonlinearization of the BO-type equations}\label{sec-5-2}

In \cite{Hie-1988} Hietarinta also investigated the following type of bilinear forms,
namely, the BO-type bilinear equations,
\begin{equation}\label{bo-def}
P(D_x,D_t) F \cdot F^* = 0,
\end{equation}
where $P$ is a binary polynomial satisfying  $[P(X,T)]^* = P(-X^*, -T^*)$.
This is a complex bilinear equation with one dependent variable $F$ and its complex conjugation.
It always has a 1-soliton solution
\begin{subequations}\label{F-BO}
\begin{equation}
F = 1 + e^{\eta},~\eta= p x + \Omega t + \eta^{(0)},
\end{equation}
$\eta = p x + \Omega t +  \eta^{(0)}$, and $ p,\Omega,\eta^{(0)}\in \mathbb{C}$
satisfying the dispersion relation
\begin{equation}
P(p,\Omega)=0.
\end{equation}
\end{subequations}

Now we come to the examples of the nonlinearization for the BO-type bilinear equations.
The first one is the bilinear BO equation \cite{Matsu-JPA-1979,Matsuno-book}:
\begin{equation}\label{BO-1}
	( D_x^2 + i D_t ) F \cdot F^* = 0,
\end{equation}
where
$F$ satisfies the conditions proposed in Proposition \ref{prop-4-3}.
Introduce
\begin{equation}\label{v,u}
v = \ln \frac{F}{F^*}, ~~ u = \ln (F F^*).
\end{equation}
On one hand,  using  formula \eqref{4.51} and relation \eqref{u_x=-iHv_x},  we have
\begin{align}
(F F^*)^{-1}( D_x^2 + i D_t ) F \cdot F^*
&=\mathcal{Y}_{[2x]}(v,u) + i \mathcal{Y}_{[t]}(v,u)  \nonumber \\
&=-i \partial_x \mathcal{H}[v_x]+v_x^2 + i v_t
= 0. \label{bo-bell}
\end{align}
On the other hand, if we further assume the solution $F$ satisfies certain condition\footnote{
There do exist solutions that fulfill such assumptions,
e.g.\cite{GHLYZ-SAPM-2026,Matsu-JPA-1979,Matsuno-book,Matsu-PLA-1986,Matsu-JMP-1991}
for multi-soliton solutions.}
such that $v_x$ and $v_{xx}$ belong to $L^p(\mathbb{R})$ space where $1< p <+\infty$,
then, in light of Proposition \ref{prop-4-1}, we obtain
\begin{equation}\label{BO-1-q}
	q_t - q_x^2 - \mathcal{H}[ q_{xx}] = 0, ~~~ (q=iv).
\end{equation}
Next, introducing $w=q_x$, we obtain the Benjamin-Ono equation
\begin{equation}
w_t - 2 w w_x - \mathcal{H} [w_{xx}] = 0.
\end{equation}
To trace back to the bilinear form \eqref{BO-1}, we have
\begin{equation}\label{w-F}
w =-i\partial_x \ln \frac{F^*}{F}.
\end{equation}

The second example is (see main results in the abstract of \cite{Hie-1988})
\begin{equation}\label{BO-2}
	(i a D_x^3 + D_x^2+ i D_t) F \cdot F^* = 0,
\end{equation}
where $a\in \mathbb{R}$,
$F$ satisfies the conditions proposed in Proposition \ref{prop-4-3}.
We employ the transformation \eqref{v,u} and still
assume  $v_x, v_{xx}\in L^p(\mathbb{R})$  where $1< p <+\infty$.
Compared with equation \eqref{BO-1}, we only need to figure out the nonlinear part generated from
$i a D_x^3   F \cdot F^*$.
Using \eqref{4.51} we have
\begin{equation}\label{ex3-bell}
i a (FF^*)^{-1} D_x^3   F \cdot F^*
= i a \mathcal{Y}_{[3x]}(v,u)
= i a ( v_{xxx} - 3i v_{x}\mathcal{H}[v_{xx}] + v_x^3).
\end{equation}
Combined with \eqref{BO-1-q}, we obtain
\begin{equation}
	q_t - q_x^2 -\mathcal{H}[q_{xx}]
+ a (q_{xxx}	- q_x^3 	-3q_{x} \mathcal{H} [q_{xx}] )= 0, ~~ (q=i v=-i\ln (F^*/F)).
\end{equation}

\begin{remark}\label{rem-1}
In Ref.\cite{Hie-1988} Hietarinta considered the BO-type bilinear equations of the form \eqref{bo-def}
that allows a ``2-soliton'' solution generalizing the form  \eqref{F-BO}.
The bilinear equation \eqref{BO-2} is one of such equations he found.
A simple solution of  \eqref{BO-2} is (cf. equation \eqref{F-BO}):
\begin{equation}
F(x)=1+e^{k x - i(  k^2+i a k^3)t +\eta^{(0)}}, ~~ k, \eta^{(0)}\in \mathbb{C}.
\end{equation}
We extend it from $x\in \mathbb{R}$ to the complex plane:
\begin{equation}\label{F-peorid}
F(z)=1+e^{k z - i(  k^2+i a k^3)t +\eta^{(0)}}, ~~ z\in \mathbb{C}.
\end{equation}
Its zeros are obtained when
\begin{equation}
k z - i(  k^2+i a k^3)t +\eta^{(0)}=i(2s+1)\pi, ~~ s\in \mathbb{Z},
\end{equation}
If we take $k=i \alpha,~ \eta^{(0)}=\gamma$ where $\alpha, \gamma\in \mathbb{R}$,
then the zeros of $F(z)$ are
\begin{equation}
z_j=(a \alpha^2 -\alpha)t +\frac{(2s+1)\pi}{\alpha}+i\frac{\gamma}{\alpha}, ~~ s\in \mathbb{Z}.
\end{equation}
which can be contained in the lower half complex plane when $\gamma/ \alpha<0$.
Thus, $F(z)$ in \eqref{F-peorid} is an entire function with all zeros $\{z_j\}$
located in the lower half complex plane in light of the above setting.
However, for such a function, it turns out that $(\ln F)_x$ is a periodic function when $k=i \alpha$,
which means it does not hold any longer for
the condition  $(\ln F)_x \to 0$ for $\mathrm{Im}\, z \geq 0$ as $|z| \to \infty$,
and consequently, such an $F(z)$ does not satisfy the condition of Proposition \ref{prop-4-4}.
In this paper, we focus on the nonlinearization of bilinear BO-type equations.
One may look for solutions of \eqref{BO-2} that satisfy the conditions in  Proposition \ref{prop-4-4}.
It is also possible to consider the periodic Hilbert transformation
(cf.\cite{King-book}, Chapter 6)
\begin{equation}\label{H-p}
	\mathcal{H} [f(x)] = \frac{1}{2L} \mathrm{P.V.} \int_{-L}^{L} \cot \left[\frac{\pi(s-x)}{2L}
\right] f(s) \mathrm{d}s, ~ ~~ f(x)=f(x+2L),
\end{equation}
and set up the propositions as in Sec.\ref{sec-5-1} and study the bilinearization/nonlinearization of the
periodic BO type equations.
We would like to investigate this elsewhere.
\end{remark}

One more example is (see main results in the abstract  of \cite{Hie-1988})
\begin{equation}
(D_x D_t + i (a D_x + b D_t)) F \cdot F^* = 0,
\end{equation}
where $a,b\in \mathbb{R}$,  $F$ is analytic in the upper half plane,
$F(z)\to 0$ as $|z|\to \infty$, and all its zeros are located in the lower half plane.
Again, we employ the transformation \eqref{v,u} and
assume  $v_x, v_{xt}$ belong to  $L^p(\mathbb{R})$ space where $1< p <+\infty$.
Then, we have
\begin{equation}\label{ex4-bell}
\mathcal{Y}_{[x,t]}(v,u)
+ i a\mathcal{Y}_{[x]}(v,u)
+ i b\mathcal{Y}_{[t]}(v,u) \\
  =- i \mathcal{H} [v_{xt}] + v_{x} v_{t}
+ i a v_x + i b v _t = 0,
\end{equation}
i.e.,
\begin{equation}
	\mathcal{H}[q_{xt}] +q_x q_t - a q_x - b q_t = 0, ~~ (q=i v=-i\ln (F^*/F)).
\end{equation}

\subsection{Other examples}\label{sec-5-3}

It is not necessary to restrict us to the bilinear equations of the form \eqref{bo-def}.
For some bilinear equations, one can also introduce  Hilbert transformations
by using Proposition \ref{prop-4-4}.
In this subsection, we always assume $F$ and $G$ satisfy the conditions mentioned in Proposition \ref{prop-4-3},
without repeating them.

Let us start from the following bilinear form:
\begin{equation}\label{sh-bi}
	(F^* F)_t = \frac{1}{2 i} ( F^2 - {F^*}^2 ).
\end{equation}
Directly employing the transformation \eqref{v,u}, we obtain
\begin{equation}\label{5.31}
	u_t =  \frac{e^{v} - e^{-v}}{2i}.
\end{equation}
We further assume  $v, v_x, v_{t}$ belong to  $L^p(\mathbb{R})$ space where $1< p <+\infty$.
This allows us to have $u=-i \mathcal{H} [v] $ from \eqref{u_x=-iHv_x}
and then $u_t=-i\mathcal{H} [v_t]$ in light of Proposition \ref{prop-4-1}.
Thus from \eqref{5.31} we obtain
\begin{equation}
	\mathcal{H} [q_t] =  \sin q, ~ ~~ (q=-i v).
\end{equation}
This is known as the sine-Hilbert equation \cite{Matsu-PLA-1986}  and is related to the bilinear form
\eqref{sh-bi} via $q = i \ln (F^*/F)$.

The next example is (cf.\cite{YLH-CPL-2024})
\begin{equation}\label{msh-bi}
D_t G^* \cdot F = \alpha ( G^* F^* - G F), ~~~\alpha \alpha^*= \frac{1}{4},
\end{equation}
where
\begin{equation}\label{msh v u}
v'=\ln \frac{G^*}{F},~~u'= \ln FG^*
\end{equation}
Equation \eqref{msh-bi} can be represented as
\begin{equation}\label{msh-bi-1}
(F G^*)^{-1} D_t G^* \cdot F
= \mathcal{Y}_{[t]} (v',u')= v'_t
= \alpha \bigg( \frac{F^*}{F} - \frac{G}{G^*} \bigg),
\end{equation}
which gives
\begin{equation}\label{msh-vt}
v'_t = \alpha e^{\frac{1}{2} ( u^{\prime *} - u') }
\left( e^{-\frac{1}{2} ( v^{\prime *} - v') } - e^{\frac{1}{2} ( v^{\prime *} - v') }\right).
\end{equation}
Multiplying it  with its complex conjugate form yields
\begin{equation}\label{msh-vt v^*t}
v'_t v^{\prime *}_t = - \frac{1}{4}
\left( e^{ v^{\prime *} - v' } + e^{-( v^{\prime *} - v') } -2 \right),
\end{equation}
where we have replaced $\alpha \alpha^*$ with $\frac{1}{4}$.
To go further, we write its  left-hand side   into
\begin{equation}\label{msh eq -left}
v'_t v^{\prime *}_t = \frac{1}{4}  [ (v_t^{\prime *} + v'_t)^2 - (v_t^{\prime *} - v'_t)^2 ].
\end{equation}
For the right-hand side of \eqref{msh-vt v^*t}, using the identity
\begin{equation}
\ln \frac{F G}{F^*G^*}
= 2i \arctan \frac{ \mathrm{Im} (F G)}{ \mathrm{Re}( F G)}
= 2i \arctan \frac{F G - F^*G^*}{ i (F G + F^*G^*)}
\end{equation}
and \eqref{msh v u}, we have
\begin{align}
e^{ v^{\prime *} - v' } + e^{-( v^{\prime *} - v') } -2
&=\frac{ [e^{\frac{1}{2} (u^{\prime *} - v^{\prime *} +u' + v')}
- e^{\frac{1}{2} (u^{\prime *}+ v^{\prime *} + u'- v')} ]^2}{e^{( u^{\prime *} + u') }} \nonumber \\
&=2 \frac{( F G - F^* G^* )^2}{2F F^* G G^*}
=2 \frac{(F G)^2 + (F^*G^* )^2}{2F F^* G G^*} -2 \nonumber \\
&=2 \cos \bigg( i \ln \frac{F G}{F^*G^*} \bigg) -2. \label{msh-eq right}
\end{align}
Introducing
\begin{equation}\label{w-vv*}
	w = i (v^{\prime *} - v' ) =i \ln \frac{F G}{F^*G^* }
\end{equation}
and assuming $\partial_t \ln F \to 0$ (for $\mathrm{Im}\,z \geq 0$)
and $\partial_t \ln G \to 0$ (for $\mathrm{Im}\,z \leq 0$) as $z$ goes to $\infty$,
it then follows from Proposition \ref{prop-4-3} that
\begin{align}
\mathcal{H} [w_t] = \mathcal{H} [i \partial_t (v^{\prime *} - v' )]
=&\mathcal{H} \bigg[i \partial_t \ln \frac{F G}{ F^*G^*} \bigg] \nonumber \\
=& \mathcal{H} \bigg[i \partial_t \ln \biggl(\frac{F}{G^*}\bigg / \frac{F^*}{G}\biggr) \bigg]
=- \partial_t \ln \frac{FF^*}{G^* G}
=( v^{\prime *} + v' )_t.
\end{align}
Then, combining \eqref{msh-vt v^*t}, \eqref{msh eq -left}, \eqref{msh-eq right} and the above equation,
we get a nonlinear equation
\begin{equation}
w_t^2 + (\mathcal{H} w_t)^2 +2 \cos w - 2 = 0.
\end{equation}
This is known as the modified sine-Hlibert equation \cite{YLH-CPL-2024}
and is related to the bilinear form \eqref{msh-bi} via the transformation \eqref{w-vv*}.

The third example is (cf.\cite{Naka-JPSJ-1979})
\begin{subequations}\label{mbo-bi}
\begin{align}
&(i D_t - 2  {i} \lambda D_x - D_x^2 - \mu ) F \cdot \tilde G = 0,\\
&(i D_t - 2  {i} \lambda D_x - D_x^2 - \mu ) \tilde{F} \cdot {G} =0,\\
&( D_x + i \gamma) F \cdot{G} - i \nu \tilde{F}\tilde G = 0
\end{align}
\end{subequations}
with $\lambda, \gamma, \mu, \nu$ being constants,
and where, apart from $F$ and $G$ as described in Proposition \ref{prop-4-3},
 $\tilde{G}$ and $\tilde{F}$ are functions
analytic in the upper half and lower half  complex plane, respectively,
while their zeros are respectively located in the lower half and upper half complex plane.
In addition, $\partial_x \ln \tilde G \to 0$ (for $\mathrm{Im}\,z \geq 0$)
and $\partial_x \ln \tilde F \to 0$ (for $\mathrm{Im}\,z \leq 0$) as $z$ goes to $\infty$.
Except for $v$ and $u$ defined as in \eqref{uv-FG}, we also introduce
\begin{equation}\label{mbo v u}
%v=\ln \frac{F}{G},
%~u=\ln FG,
\tilde{v}=\ln \frac{F}{\tilde{G}},
~~\tilde{u}=\ln F \tilde{G},
~~\tilde{\tilde{v}}=\ln \frac{\tilde{F}}{G},
~~\tilde{\tilde{u}}=\ln \tilde{F}G.
\end{equation}
Using the formula \eqref{trans-formula}, we can convert the bilinear system \eqref{mbo-bi} into the following form
\begin{subequations}\label{mbo-bi-1}
\begin{align}
&i \mathcal{Y}_{[t]}(\tilde v,\tilde u) - 2 i \lambda \mathcal{Y}_{[x]}(\tilde v, \tilde u)
- \mathcal{Y}_{[2x]}(\tilde v, \tilde u) - \mu =0,\\
&i \mathcal{Y}_{[t]}(\tilde{\tilde{v}},\tilde{\tilde{u}}) - 2 i \lambda \mathcal{Y}_{[x]}(\tilde{\tilde{v}},\tilde{\tilde{u}})
- \mathcal{Y}_{[2x]}(\tilde{\tilde{v}},\tilde{\tilde{u}}) - \mu =0,\\
& \mathcal{Y}_{[x]}({v}, {u}) + i \gamma = i \nu \frac{\tilde{F}\tilde G}{G F}.
\end{align}
\end{subequations}
These lead to
\begin{subequations}\label{mbo-bi-2}
\begin{align}
&i \tilde v_t -2 i \lambda \tilde v_x - (\tilde u_{xx} + \tilde v_x^2) - \mu = 0,\label{mbo-a}\\
&i \tilde{\tilde{v}}_t -2 i \lambda \tilde{\tilde{v}}_x
- ( \tilde{\tilde{u}}_{xx} + (\tilde{\tilde{v}}_x)^2) - \mu = 0,\label{mbo-b}\\
&  ( \tilde u + \tilde v - \tilde{\tilde{u}} + \tilde{\tilde{v}} )_x + 2 i \gamma =2  i \nu e^{\tilde{\tilde{v}}-\tilde v}, \label{mbo-c}
\end{align}
\end{subequations}
where for the last equation we have used the relation $\tilde u + \tilde v - \tilde{\tilde{u}} + \tilde{\tilde{v}}=-2v$.
Introduce
\begin{equation}\label{mbo w fg}
w =\tilde{\tilde{v}}-\tilde v = \ln \frac{\tilde{F}\tilde G}{{G} F}.
\end{equation}
Then, subtracting \eqref{mbo-b} from \eqref{mbo-a} and differentiating \eqref{mbo-c} with respect to $x$
yield
\begin{subequations}
\begin{align}
& i w_t -2 i \lambda w_x - (\tilde{\tilde{u}} -\tilde u)_{xx}
 - w_x ( \tilde{\tilde{v}} +\tilde v)_x = 0,\label{mbo-2-a}\\
&  (\tilde{\tilde{v}} +\tilde v)_{xx} - (\tilde{\tilde{u}} - u)_{xx}   =  2 i\nu w_x e^w.\label{mbo-2-b}
\end{align}
\end{subequations}
Eliminating $(\tilde{\tilde{u}} - u)_{xx}$ yields
\begin{equation}\label{mbo-3}
i w_t - 2 i \lambda w_x - ( \tilde{\tilde{v}} +\tilde v)_{xx}
+ 2 i \nu w_x e^w - w_x ( \tilde{\tilde{v}} +\tilde v)_x = 0.
\end{equation}
Note that $\mathcal{H}$ is a linear operator, which means
\begin{equation}\label{hilbert fg-1}
\mathcal{H} [w_x] = \mathcal{H} [( \tilde{\tilde{v}}-\tilde v)_x]
= \mathcal{H} \bigg[ \partial_x \ln \frac{\tilde{F}\tilde G}{{G} F} \bigg]
= \mathcal{H} \bigg[ \partial_x \ln \frac{\tilde{F}}{F} \bigg]
   + \mathcal{H} \bigg[ \partial_x \ln \frac{\tilde{G}}{G} \bigg].
\end{equation}
Making use of formula \eqref{hil-eq1} in Proposition \ref{prop-4-3} we have
\begin{equation}\label{hilbert fg-2}
\mathcal{H} [w_x] =-i\partial_x \ln(F{\tilde{F})+i\partial_x \ln(  G \tilde{G})}
=  i \partial_x \ln \frac{\tilde{G} G} {\tilde{F} F}=- i (\tilde{\tilde{v}} + \tilde v)_x.
\end{equation}
Further, we assume  $\tilde v_x, \tilde v_{xx}, \tilde{\tilde{v}}_x, \tilde{\tilde{v}}_{xx}$
belong to  $L^p(\mathbb{R})$ space where $1< p <+\infty$
so that from the above relation we have $\mathcal{H} [w_{xx}] =- i (\tilde{\tilde{v}} +\tilde v)_{xx}$.
Thus \eqref{mbo-3} yields
\begin{equation}
w_t - 2 \lambda w_x + 2 \nu e^w w_x
- \mathcal{H} [w_{xx}] - w_x \mathcal{H} [w_x]  =0,
\end{equation}
which is known as the modified BO equation \cite{Naka-JPSJ-1979}
and is connected with the bilinear form \eqref{mbo-bi} via the transformation \eqref{mbo w fg}.

\section{Concluding remarks}\label{sec-6}

This paper is a continuation of our previous research \cite{ZLZ-CTP-2025}
on the nonlinearization of bilinear equations.
In this paper, we have illustrated how the sG-type and NLS-type bilinear equations
are converted to their nonlinear forms.
We also explored the BO-type equations that contain the Hilbert transformations.
We have introduced some properties of the Hilbert transformations,
which are useful in the nonlnearization/bilinearization of the  BO-type equations.
In our procedure,
solutions (i.e. $F, G$, etc) of the bilinear equations of this type
are required to satisfy special properties for the analytic area and location of zeros.
Sometimes we also need the involved functions are ``good'' enough so that
their Hilbert transformation commutes with differential operations, e.g.
$\partial_x \mathcal{H} [v]= \mathcal{H} [v_x]$.
Note that in practice these requirements can be satisfied in construction of soliton solutions
(which are called rational solutions in \cite{Hie-1988}),
e.g.\cite{GHLYZ-SAPM-2026,Matsu-JPA-1979,Matsuno-book,Matsu-PLA-1986,Matsu-JMP-1991}.
The understanding of these properties of the Hilbert transformations
is useful not only in the nonlnearization of the BO-type bilinear equations
but also in finding bilinear forms for the nonlinear equations with Hilbert transformations.
In addition, there are periodic solutions (also known as multiphase solutions), e.g.
\cite{Naka-JPSJ-1979,DK-MN-1991,GHLYZ-SAPM-2026,SI-JPSJ-1979}.
These solutions do not satisfy the condition for $F$ and $G$ given in Proposition \ref{prop-4-3}.
We have provided an example and given some comments in Remark \ref{rem-1}.
For the nonlinearizarion related to the periodic Hilbert transformation,
we would like to investigate it elsewhere.

In this paper, we did not consider the NLS-type equations with Hilbert transformations,
e.g.
\begin{equation}
i \psi_t +\alpha \mathcal{H} [\psi_x] +\beta |\psi|^2\psi=0
\end{equation}
which has interesting physical contexts \cite{GMCR-PRE-1997,HWQL-2024}, where $\alpha$ and $\beta$ are constants.
Bilinear method has been applied to some NLS-type  equations with Hilbert transformations \cite{Matsu-PLA-2000,Matsu-SAPM-2023,LYHZ-Non-2025}.
Note that no matter whether the bilinear equations are integrable or not,
the sG-type bilinear system \eqref{sG-bilinear-0} always has a 1-soliton solution and 2-soliton solution,
the NLS-type bilinear system \eqref{4.1}  always has a 1-soliton solution
and the BO-type bilinear equation \eqref{bo-def} always has a 1-soliton solution.
This means it is possible to use bilinear method to investigate non-integrable equations.
We believe our research in this paper as well as in our previous one \cite{ZLZ-CTP-2025}
on the nonlinearization of bilinear equations can bring new understanding
in the study of various types of models from the bilinear approach \cite{Hir-PRL-1971,Hirota-book}.

\vskip 20pt
\subsection*{Acknowledgments}

DJZ is supported by the Natural Science Foundation  of China (No.12271334).
XHZ is supported by the Youth Science Fund (Type B) (No.2026QB044) and the Key Project (No.2026ZD036) of the Natural Science Foundation of Inner Mongolia Autonomous Region, the First-class Discipline Research Project of Inner Mongolia Normal University (No.YLXKZX-NSD-001 and YLXKZX-NSD-009) and Program for Innovative Research Team in Universities of Inner Mongolia Autonomous Region (No.NMGIRT2414), China.

\end{document}